\documentclass[conference]{IEEEtran}
\ifCLASSINFOpdf
\else
\fi
\usepackage{url}

\usepackage{hyperref}

\usepackage{amsmath}
\usepackage{amsfonts}
\usepackage{amssymb}
\usepackage{amsthm}

\usepackage{enumerate}

\usepackage{bbold}

\usepackage{graphicx}
\usepackage{caption}%

\newtheorem{definition}{Definition}

\newtheorem{theorem}{Theorem}
\newtheorem{lemma}{Lemma}
\newtheorem{corollary}{Corollary}
\newtheorem{claim}{Claim}
\newtheorem{remark}{Remark}

\usepackage[backend=bibtex,style=alphabetic,maxnames=5,minnames=5,maxbibnames=99,firstinits=true,
doi=false,isbn=false,url=false
]{biblatex}
\AtEveryBibitem{%
    \clearfield{pages}
}

\usepackage[linesnumbered,algoruled,boxed,lined]{algorithm2e}
\SetCommentSty{medium}

\renewcommand\IEEEkeywordsname{Keywords}

\IEEEoverridecommandlockouts

\usepackage[nameinlink]{cleveref}
\Crefname{appsec}{Appendix}{Appendices}
\Crefname{theorem}{Theorem}{Theorems}
\Crefname{section}{Section}{Sections}
\Crefname{algorithm}{Algorithm}{Algorithms}
\Crefname{table}{Table}{Tables}

\begin{document}
%
\title{On the Complexity of Estimating Renyi Divergences}

\author{\IEEEauthorblockN{Maciej Skorski\thanks{Supported by the European Research Council consolidator grant (682815-TOCNeT).}}
\IEEEauthorblockA{IST Austria \\ Email: mskorski@ist.ac.at}
}


%


\maketitle

\begin{abstract}
This paper studies the complexity of estimating Renyi divergences of discrete distributions: $p$ observed from samples and the baseline distribution $q$ known \emph{a priori}. 
Extending the results of Acharya et al. (SODA'15) on estimating Renyi entropy, we present improved estimation techniques together with upper and lower bounds on the sample complexity.

We show that, contrarily to estimating Renyi entropy where a sublinear (in the alphabet size) number of samples suffices, 
the sample complexity is heavily dependent on \emph{events occurring unlikely} in $q$, and is unbounded in general
(no matter what an estimation technique is used). 
For any divergence of order bigger than $1$,  we provide upper and lower bounds on the number of samples dependent on probabilities of $p$ and $q$.
We conclude that the worst-case sample complexity is polynomial in the alphabet size if and only if the probabilities of $q$ are non-negligible.

This gives theoretical insights into heuristics used in applied papers to handle
numerical instability, which occurs for small probabilities of $q$. Our result explains that
small probabilities should be handled with care not only because of numerical issues, but also because
of a blow up in sample complexity.

\end{abstract}


\vspace{0.2cm}


%
\IEEEpeerreviewmaketitle

\section{Introduction}

\subsection{Renyi Divergences in Anomaly Detection}

A popular statistical approach to detect anomalies in real-time data is to compare
the empirical distribution of certain features (updated on the fly) against
a stored ``profile'' (learned from past observations or computed off-line) used as a reference distribution. 
Significant deviations of the observed distribution from the assumed profile trigger an alarm~\cite{Gu:2005:DAN:1251086.1251118}.

This technique, among many other applications, is often used 
to detect DDoS attacks in network traffic~\cite{DBLP:journals/eswa/GulisanoCFJPP15,DBLP:conf/isncc/PukkawannaKY15}.
To quantify the deviation between the actual data and the reference distribution,
one needs to employ a suitable dissimilarity metric. In this context, based on empirical studies, Renyi divergences were suggested as good 
dissimilarity measures~\cite{DBLP:journals/iet-com/LiZY09,DBLP:journals/tifs/XiangLZ11,DBLP:journals/prl/BhuyanBK15,DBLP:journals/eswa/GulisanoCFJPP15,DBLP:conf/isncc/PukkawannaKY15}.

While the divergence can be evaluated based on theoretical models\footnote{
For example, one uses fractional Brownian motions to simulate real network traffic
and Poisson distributions to model DDoS traffic\cite{DBLP:journals/tifs/XiangLZ11}.},
much more important (especially for real-time detection) is the estimation on the basis of samples.
The related literature is focused mainly on tunning the performance
of specific implementations, by choosing appropriate parameters (such as the suitable definition or the sampling frequency)
based on empirical evidence. 
On the other hand, not much is known about the theoretical performance of estimating Renyi divergences for general discrete distributions
(continuous distributions need extra smoothness assumptions~\cite{DBLP:conf/icml/SinghP14}). 
A limited case is estimating Renyi entropy~\cite{DBLP:conf/soda/AcharyaOST15} which corresponds
to the uniform reference distribution.

In this paper, we attempt to fill the gap by providing 
better estimators for the Renyi divergence, together with theoretical guarantees on the performance. In our approach, motivated
by mentioned applications to anomaly detection, we assume that the reference distribution $q$ is explicitly known and
the other distribution $p$ can only be observed from i.i.d. samples.

\subsection{Our Contribution and Related Works}


\paragraph{Better estimators for a-priori known reference distributions}

In the literature Renyi divergences are typically estimated by straightforward \emph{plug-in} estimators
(see \cite{DBLP:journals/iet-com/LiZY09,DBLP:journals/prl/BhuyanBK15,DBLP:journals/iet-com/LiZY09,DBLP:journals/tifs/XiangLZ11,DBLP:journals/prl/BhuyanBK15,DBLP:journals/eswa/GulisanoCFJPP15,DBLP:conf/isncc/PukkawannaKY15}).
In this approach, one puts the empirical distribution (estimated from samples)
into the divergence formula, in place of the true distribution. Unfortunately, they have worse statistical properties, e.g.
are heavily biased. This affects the number of samples required to get a reliable estimate.

To obtain reliable estimates within a possible small number of samples, we extend the techniques
from \cite{DBLP:conf/soda/AcharyaOST15}. The key idea is to use \emph{falling powers} to estimate
power sums of a distribution (this trick is in fact a bias correction method).
The estimator is illustrated in \Cref{alg:estimator} below.


\begin{algorithm}[t]
\DontPrintSemicolon
  \KwIn{divergence parameter $\alpha > 1$, \newline
  alphabet $\mathcal{A} = \{a_1,\ldots,a_k\}$, \newline
  reference distribution $q$ over $\mathcal{A}$, \newline
  samples $x_1,\ldots,x_n$ from unknown $p$ on $\mathcal{A}$ \\ }
  \KwOut{a number $D$  approximating the $\alpha$-divergence}    
initialize\footnotemark $n_i=0$ for all $i=1,\ldots,k$ \\
\For{$j=1,\ldots,n$}{
let $i$ be such that $x_j = a_i$ \\
$n_i \gets n_i+1$  
\tcc*{compute empirical frequencies} 
} 
$M\gets \sum_{i} q_i^{1-\alpha}\cdot \frac{ n_i^{\underline{\alpha}} }{n^{\alpha}}$ \tcc*{bias-corrected power sum estimation, $z^{\underline{\alpha}}$ stands for 
the falling $\alpha$-power.}
$D\gets \frac{1}{\alpha-1}\log M$ \tcc*{divergences from power sums} 
\Return{$D$}
\caption{Estimation of Renyi Divergence (to a reference distribution known a-priori)}\label{alg:estimator}
\end{algorithm}
\footnotetext{Storing and updating empirical frequencies can be implemented more efficiently when $n \ll k$, which
matters for almost uniform distributions $q$ (sublinear time and memory complexity), but not for the general case.}

For certain cases (where the reference distribution is close to uniform) we estimate the divergence
with the number of samples \emph{sublinear in the alphabet size}, whereas plug-in estimators need a
superlinear number of samples. In particular for the uniform reference distribution $q$, we recover the same upper bounds for 
\emph{estimating Renyi entropy} as in \cite{DBLP:conf/soda/AcharyaOST15}.

\paragraph{Upper and lower bounds on the sample complexity}

We show that the sample complexity of estimating divergence of unknown $p$ observed from samples
to an explicitly known $q$ is dependent on the reference distribution $q$ itself.
When $q$ doesn't take too small probabilities, non-trivial estimation is possible, even \emph{sublinear} in the alphabet size for any $p$.
However when $q$ takes arbitrarily small values, the complexity is dependent on inverse powers of probability masses of $p$ and
is \emph{ unbounded} (for a fixed alphabet), without extra assumptions on $p$. We stress that these lower bounds are {no-go results} independent of the estimation technique.
For a more quantitative comparison, see \Cref{table:comparison}.

\begin{table}[h]
\captionsetup{skip=5pt,font=small}
\resizebox{0.48\textwidth}{!}{
\begin{tabular}{|l|l|l|l|}
\hline
 Assumption & Complexity & Comment & Reference \\
 \hline
 $\min_i q_i = \Theta(k^{-1}) = \max_{i}q_i$ & $O\left( k^{1-\frac{1}{\alpha}}\right)$ & \parbox{3.5cm}{almost uniform $q$, \\ complexity sublinear} & \Cref{cor:sublinear_complexity}\\
 \hline
 no assumptions & $\Omega\left(k^{\frac{1}{2}}\right)$ & \parbox{3.5cm}{complexity at least square root} & \Cref{cor:lower_bounds_gen}\\
 \hline
 $\min_i q_i = k^{-\omega(1)}$ & $k^{\omega(1)}$ & \parbox{3.5cm}{negligible masses in $q$, \\ super-polynomial complexity} & \Cref{cor:sample_blowup} \\
\hline
  $\min_i q_i = k^{-O(1)}$ & $k^{O(1)}$ & \parbox{3.5cm}{non-negligible mass in $q$, \\ polynomial complexity} &
\Cref{cor:poly_complexity}
  \\
  \hline
\end{tabular}
}
\caption{{A brief summary of our results, for the problem of estimating the Renyi divergence $D_{\alpha}(p\parallel q)$ (where 
the divergence parameter $\alpha>1$ is a fixed constant)
between the known baseline distribution $q$ and a distribution $p$ learned from samples, both over an alphabet of size $k$.
The complexity is the number of samples needed to estimate the divergence up to a constant error
and with success probability at least $\frac{2}{3}$}.
}
\label{table:comparison}
\end{table}

\paragraph{Complexity instability vs numerical instability}

Our results provide theoretical insights about heuristic ``patches'' to the Renyi divergence formula
suggested in the applied literature. Since the formula is numerically unstable when
one of the probability masses $q_i$ becomes arbitrarily small (see \Cref{dif:renyi_div}), authors suggested
to omit or round up very small probabilities of $q$ (see for example \cite{DBLP:journals/iet-com/LiZY09,DBLP:conf/isncc/PukkawannaKY15}).

In accordance to this, as shown in \Cref{table:comparison}, the sample complexity is also unstable when unlike events occur in the reference distribution $q$.
Moreover, this is the case even if the distribution $q$ is perfectly known.
We therefore conclude that small probabilities of $q$ are very subtle not only because
of numerical instability, but more importantly because the sample complexity is unstable.

%


\subsection{Our techniques}

For upper bounds we merely borrow and extend techniques from \cite{DBLP:conf/soda/AcharyaOST15}.
For lower bounds our approach is however different. 
We find a pair of distributions which are close in total variation yet with much different divergences to $q$,
by a variational approach (writing down an explicit optimization program) 
As a result, we can obtain our lower bounds for any accuracy.
In turn, the argument in \cite{DBLP:conf/soda/AcharyaOST15}, even if can be extended to the Renyi divergence,
has inherit limitations that make it work only for sufficiently small accuracies.
Thus we can say that our lower bound technique, in comparison to \cite{DBLP:conf/soda/AcharyaOST15}, offers
lower bounds valid in all regimes of the accuracy parameter, in particular for constant values used in the applied literature.

In fact, our technique strictly improves known lower bounds on estimating collision entropy.
Taking the special case when $q$ is uniform, we obtain that the sample complexity for estimating collision entropy
is $\Omega(k^{\frac{1}{2}})$ even for constant accuracy, while results in \cite{DBLP:conf/soda/AcharyaOST15} guarantees
this only for very small $\delta$ (no exact threshold is given, and hidden constants may be dependent on $\delta$), which is captured by the notation $\tilde{\tilde{\Omega}}(k^{\frac{1}{2}})$.

\subsection{Organization}

In \Cref{sec:prelim} we introduce necessary notions and notations. 
Upper bounds on the sample complexity are discussed in \Cref{sec:upper_bounds} and lower bounds
in \Cref{sec:lower_bounds}. We conclude our work in \Cref{sec:conclusion}.

\section{Preliminaries}\label{sec:prelim}

For a distribution $p$ over an alphabet $\mathcal{A} = \{a_1,\ldots,a_k\}$ we denote $p_i = p(a_i)$. All logarithms are at base $2$.

\begin{definition}[Total variation]
The total variation of two distributions $p,p'$ over the same finite alphabet equals
$d_{TV}(p,p') = \frac{1}{2}\sum_{i}|p_i-p'_i|$.
\end{definition}

Below we recall the definition of Renyi divergence (we refer the reader to \cite{DBLP:journals/tit/ErvenH14} for a survey of its properties).

\begin{definition}[Renyi divergence]\label{dif:renyi_div}
The Renyi divergence of order $\alpha$ (in short: Renyi $\alpha$-divergence) of two distributions $p,q$ having the same support is defined by
\begin{align}\label{eq:renyi_div}
D_{\alpha}(p\parallel q)= \frac{1}{\alpha-1}\log\sum_{i}\frac{p_i^{\alpha}}{q_i^{\alpha-1}}
\end{align}
\end{definition}
By setting uniform $q$ we get the relation to Renyi entropy.
\begin{remark}[Renyi entropy vs Renyi divergence]\label{rem:entropy_from_divergence}
For any $p$ over $\mathcal{A}$ the Renyi entropy of order $\alpha$ equals  
\begin{align*}
 -\frac{1}{\alpha-1}\log\sum_{i}{p_i^{\alpha}} = -D_{\alpha}(p\parallel q_{\mathcal{A}}) + \log |\mathcal{A}|,
\end{align*}
where $q_{\mathcal{A}}$ is the uniform distribution over $\mathcal{A}$.
\end{remark}

\begin{definition}[Renyi's divergence estimation]\label{def:estimator}
Fix an alphabet $\mathcal{A}$ of size $k$, and two distributions $p$ and $q$ over $\mathcal{A}$.
Let $\mathsf{Est}^{q}:\mathcal{A}^n\rightarrow \mathbb{R}$ be an algorithm which receives $n$ independent samples of $p$ on its input.
We say that $\mathsf{Est}^{q}$ provides an additive $(\delta,\epsilon)$-approximation to the Renyi $\alpha$-divergence of $p$ from $q$ if
\begin{align}\label{eq:estimator}
 \Pr_{x_i\gets p}\left[ | \mathsf{Est}^{q}(x_1,\ldots,x_n) - D_{\alpha}(p\parallel q)| > \delta \right] < \epsilon.
\end{align}
\end{definition}
\begin{definition}[Renyi's divergence estimation complexity]\label{def:complexity}
The sample complexity of estimating the Renyi divergence given $q$ with probability error $\epsilon$ and
additive accuracy $\delta$ is the minimal number $n$ for which there exists an algorithm satisfying \Cref{eq:estimator}
for all $p$.
\end{definition}

It turns out that it is very convenient not to work directly with estimators for Renyi divergence,
but rather with estimators for weighted power sums.
\begin{definition}[Divergence power sums]
The power sum corresponding to the $\alpha$ divergence of $p$ and $q$ is defined as
\begin{align}\label{eq:renyi_moments}
 M_{\alpha}(p,q) \overset{def}{=} \mathrm{e}^{(1-\alpha)D_{\alpha}(p\parallel q)=}= \sum_{i} \frac{p_i^{\alpha}}{q_i^{\alpha-1}}
\end{align}
\end{definition}
The following lemma shows that estimating divergences (\Cref{eq:renyi_div}) with an absolute relative error of $O(\delta)$ and
corresponding power sums (\Cref{eq:renyi_moments}) with a relative error of $O(\delta/(\alpha-1))$ is equivalent
\begin{lemma}[Equivalence of Additive and Multiplicative Estimations]\label{lemma:div_to_moments}
Suppose that $m$ is a number such that $M_{\alpha}(p,q) = m\cdot (1+\delta)$, where $|\delta| < \frac{1}{2}$.
Then $d = -\frac{1}{\alpha-1}\log m$ satisfies 
$D_{\alpha}(p\parallel q) = d+O(1/(\alpha-1))\cdot \delta$.
The other way around, if $m'$ is such that $D_{\alpha}(p\parallel q) = d+\delta$, where $|\delta| < \frac{1}{2}$, then
$m = \mathrm{e}^{(1-\alpha)d}$ satisfies $M_{\alpha}(p,q) = m\cdot (1+ O(\alpha-1)\cdot \delta )$.
\end{lemma}
The proof is a straightforward consequence of the first order Taylor's approximation, and will appear in the full version.

\section{Upper Bounds on the Sample Complexity}\label{sec:upper_bounds}

Below we state our upper bounds for the sample complexity. The result is very similar to the formula in \cite{DBLP:conf/soda/AcharyaOST15} before simplifications,
except the fact that in our statement there are additional weights coming from possibly non-uniform $q$ and
it can't be further simplified. 
\begin{theorem}[Generalizing \cite{DBLP:conf/soda/AcharyaOST15}]\label{thm:upper_bounds}
For any distributions $p,q$ over an alphabet of size $k$, if the number $n$ satisfies
\begin{align*}
\sum_{r=0}^{\alpha-1}\binom{\alpha}{r}\frac{1}{n^{\alpha-r}} \frac{\sum_{i}\frac{p_i^{\alpha+r}}{q_i^{2\alpha-2}}}{\left(\sum_{i}\frac{p_i^{\alpha}}{q_i^{\alpha-1}}\right)^2} 
\ll \epsilon\delta^2,
\end{align*}
then the complexity of estimating the Renyi $\alpha$-divergence of $p$ to given $q$ 
is at most $n$.
\end{theorem}

The proof is deferred to the appendix, below we discuss corollaries.
The first corollary shows that the complexity is sublinear when the reference distribution is close to uniform.

\begin{corollary}[Sublinear complexity for almost uniform reference probabilities, extending \cite{DBLP:conf/soda/AcharyaOST15}]\label{cor:sublinear_complexity}
Let $p,q$ be distributions over an alphabet of size $k$, and $\alpha>1$ be a constant. Suppose
that $\max_{i} q_i = O(k^{-1})$ and $\min_{i}q_i =\Omega(k^{-1})$.
Then the complexity of estimating the Renyi $\alpha$-divergence with respect to $q$, up to constant accuracy and probability error at most $\frac{1}{3}$,
is $O\left(k^{\frac{\alpha-1}{\alpha}}\right)$. 
\end{corollary}

As shown in the next corollary, the complexity is polynomial only if the reference probabilities are not negligible.

\begin{corollary}[Polynomial complexity for non-negligible reference probabilities]\label{cor:poly_complexity}
Let $p,q$ be distributions over an alphabet of size $k$.
Suppose that $\min_{i}{q_i} = k^{-O(1)}$, and let $\alpha > 1$ be a constant. Then the complexity of estimating the Renyi $\alpha$-divergence with respect to $q$, up to a constant accuracy and probability error 
at most $0.3$ (in the sense of \Cref{def:complexity}) is $k^{O(1)}$.
\end{corollary}

\begin{proof}[Proof of \Cref{cor:poly_complexity}]
Under our assumptions $\sum_{i}\frac{p_i^{\alpha+r}}{q_i^{2\alpha-2}} = k^{O(1)}\cdot \sum_{i}p_i^{\alpha+r}$.
Since $q_i \leqslant 1$, we get $\sum_{i}\frac{p_i^{\alpha}}{q_i^{\alpha-1}} \leqslant \sum_{i}p_i^{\alpha}$. By
\Cref{thm:upper_bounds}, we conclude that the sufficient condition is
\begin{align*}
\frac{\sum_{i}\frac{p_i^{\alpha+r}}{q_i^{2\alpha-2}}}{\left(\sum_{i}\frac{p_i^{\alpha}}{q_i^{\alpha-1}}\right)^2} =
k^{O(1)}\cdot \frac{\sum_{i}p_i^{\alpha+r}}{\left(\sum_{i}p_i^{\alpha}\right)^2}.
\end{align*}
Therefore, we need to chose $n$ such that
\begin{align*}
k^{O(1)}\cdot \sum_{r=0}^{\alpha-1}\binom{\alpha}{r}\frac{1}{n^{\alpha-r}}\frac{\sum_{i}p_i^{\alpha+r}}{\left(\sum_{i}p_i^{\alpha}\right)^2} < 
0.3.
\end{align*}
By the discussion in 
\cite{DBLP:conf/soda/AcharyaOST15} we know that for $r=0,\ldots,\alpha-1$ we have
$\frac{\sum_{i}p_i^{\alpha+r}}{\left(\sum_{i}p_i^{\alpha}\right)^2} \leqslant k^{(\alpha-1)\cdot \frac{\alpha-r}{\alpha}}$.
Thus we need to find $n$ that satisfies
\begin{align*}
 k^{O(1)}\cdot \sum_{r=0}^{\alpha-1}\binom{\alpha}{r} \left(\frac{n}{k^{\frac{\alpha-1}{\alpha}}}\right)^{r-\alpha}
 < 0.3,
\end{align*}
By the inequality $\sum_{j\geqslant 0}\binom{\beta}{j}u^{j} \leqslant (1+u)^{\beta}$ 
(which follows by the Taylor's expansion for any positive real number $\beta$)
and the symmetry of binomial coefficients we need
\begin{align*}
 k^{O(1)}\cdot \left(\left(1+\frac{k^{\frac{\alpha-1}{\alpha}}}{n} \right)^{\alpha}-1\right) < 0.3
\end{align*}
By the Taylor expansion $(1+u)^{\alpha} = 1+O(\alpha u)$ valid for $u\leqslant \frac{1}{\alpha}$ it suffices if 
\begin{align*}
k^{O(1)}\cdot \alpha \cdot \frac{k^{\frac{\alpha-1}{\alpha}}}{n} < 0.3.
\end{align*}
which finishes the proof.
\end{proof}

\begin{proof}[Proof of \Cref{cor:sublinear_complexity}]
The corollary can be concluded by inspecting the proof of \Cref{cor:poly_complexity}. 
The bounds are the same except that the factor $k^{O(1)}$ is replaced by $\Theta(1)^{\alpha}$.
For constant $\alpha$, the final condition reduces to $n \geqslant \Omega\left(k^{\frac{\alpha-1}{\alpha}}\right)$.
\end{proof}


\section{Sample Complexity Lower Bounds}\label{sec:lower_bounds}

The following theorem provides lower bounds on the sample complexity for any distribution
$p$ and $q$. Since the statement is somewhat technical,
we discuss only corollaries and refer to the appendix for a proof.

\begin{theorem}[Sample Complexity Lower Bounds]\label{thm:lower_bounds}
Let $p,q$ be two fixed distributions, $\delta \in (0,0.5)$ and numbers $C_1,C_2 \geqslant 0$ be given by
\begin{align*}
C_1 & = \alpha\frac{\sum_i \delta_i p_i^{\alpha}q_i^{1-\alpha}}{\sum_i p_i^{\alpha}q_i^{1-\alpha}} \cdot \delta^{-1} \\
C_2 & = \frac{\alpha(\alpha-1)}{4}\frac{\sum_i \delta_i^2 p_i^{\alpha}q_i^{1-\alpha}}{\sum_i p_i^{\alpha}q_i^{1-\alpha}} \cdot \delta^{-2}
\end{align*}
for some $\delta_i$ satisfying $\delta_i \geqslant -\frac{1}{2}$, $\sum_i \delta_i p_i = 0$, and
$\sum_{i}p_i|\delta_i|  =\delta$.
Then for any fixed $\alpha>1$, estimating the Renyi divergence to $q$ (in the sense of \Cref{def:estimator}) with error probability $\frac{1}{3}$ and up to a constant accuracy 
requires is at least
\begin{align*}
n = \Omega\left( \max(\sqrt{C_2},C_1)\right)
\end{align*}
samples from $p$.
\end{theorem}

By choosing appropriate numbers in \Cref{thm:lower_bounds} we can obtain bounds for different settings.

\begin{corollary}[Lower bounds for general case]\label{cor:lower_bounds_gen}
Estimating the Renyi divergence requires always $\Omega\left( k^{\frac{1}{2}} \right)$ samples.
\end{corollary}
\begin{proof}[Proof of \Cref{cor:lower_bounds_gen}]
In \Cref{thm:lower_bounds} we chose the uniform $p$ and $\delta$ such that
$\delta_i = \frac{k}{4}$ for the index $i=i_0$ minimizing $q_i$, and $\delta_i = -\frac{k}{4(k-1)}$ elsewhere. This gives 
us $C_1\geqslant 0$ and $C_2 \geqslant \Omega(k^2)\cdot \frac{q_{i_0}^{1-\alpha}}{\sum_{i}q_i^{1-\alpha}}$ (the constant dependent on $\alpha$) which is bigger than $\Omega(k)$,
because $\frac{q_{i_0}^{1-\alpha}}{\sum_{i}q_i^{1-\alpha}} \geqslant k^{-1}$ by our choice of $i_0$.
\end{proof}

\begin{corollary}[Polynomial complexity requires non-negligible probability masses]\label{cor:sample_blowup}
For sufficiently large $k$, if $\min_i q_i = k^{-\omega(1)}$ then there exists a distribution $p$ dependent on $k$ such that 
estimation is at least $k^{\omega(1)}$.
\end{corollary}
\begin{proof}[Proof of \Cref{cor:sample_blowup}]
Fix one alphabet symbol $a_{i_0}$ and real positive numbers $c,d$.
Let $q$ put the probability $\frac{1}{k^{c}}$ on $x$ and be uniform elsewhere. Also
let $p$ put the probability $\frac{1}{k^{d}}$ on $x$ and be uniform elsewhere. We have
\begin{align*}
 \frac{p_i^{\alpha}}{q_i^{\alpha-1}}  = \left\{
 \begin{array}{cc}
  O(k^{-1}) & i\not=i_0 \\
  k^{c(\alpha-1)-d\alpha} & i = i_0
 \end{array}
\right.
\end{align*}
and 
\begin{align*}
 \max_{i}\frac{p^{\alpha-2}}{q^{\alpha-1}_i} &=\max( k^{c(\alpha-1)-d(\alpha-2)},1)
\end{align*}
Choose $d$ so that it satisfies
\begin{align*}
c(\alpha-1)-d\alpha & > -1 \\
c(\alpha-1)-d(\alpha-2) & > 0
\end{align*}
for example
$d = \frac{\alpha-1}{\alpha}\cdot c$ (works for $\alpha \geqslant 2$ and $1< \alpha < 2$)  we obtain from \Cref{thm:lower_bounds} 
(where we take $\delta_i = \frac{1}{2}$ for $i=i_0$ and constant $\delta_i $ elsewhere,
and our conditions on $d$ ensure that $C_1\geqslant 0$ and $C_2 \geqslant \Omega(k^{2d})$ respectively
)
that for sufficiently large $k$ the minimal number of samples is 
\begin{align*}
 n = \Omega\left( k^{d} \right).
\end{align*}
Note that if $c=\omega(1)$ our choice of $d$ implies that also $d=\omega(1)$, and thus the corollary follows.
\end{proof}

\section{Conclusion}\label{sec:conclusion}

We extended the techniques recently used to analyze the complexity of entropy estimation
to the problem of estimating Renyi divergence. We showed
that in general there are no uniform bounds on the sample complexity,
and the complexity is polynomial in the alphabet size if and only if the reference distribution 
doesn't take negligible probability masses (explained by the numerical properties of the divergence formula).






%

\printbibliography

\appendix

\crefalias{section}{appsec}

\section{Proof of \Cref{thm:upper_bounds}}
\begin{proof}[Proof of \Cref{thm:upper_bounds} (sketch)]
We follow essentially the same proof strategy as in \cite{DBLP:conf/soda/AcharyaOST15}, with the only 
difference that we estimate weighted power sums $\sum_{i}q_i^{1-\alpha} {p_i^{\alpha}}$ corresponding to the divergence,
instead of sums $\sum_{i}p_i^{\alpha}$ corresponding to the entropy. 
Let $\hat{p_i}$ be the empirical frequency of the $i$-th symbol in the stream $X_1,\ldots,X_n$.
Consider the following estimator for $2^{(\alpha-1)D_{\alpha}(p,q)}$.
\begin{align*}
 M^{\mathsf{Est}}_{\alpha}(p,q) = \frac{1}{n^{\alpha}}\sum_{i} (n\hat{p_i})^{\underline{\alpha}} q_i^{1-\alpha} 
\end{align*}
Note that this is precisely the power sum defined in \Cref{alg:estimator}. By \Cref{lemma:div_to_moments}
it suffices to consider this estimator with the multiplicative error $O(\delta)$ (for constant $\alpha$).

In particular, 
we use the fact that we can randomize $n$ and make it a sample from the Poisson distribution of the same mean.
This transformation doesn't hurt the estimator convergence, but on the other hand makes 
the empirical frequencies independent (see \cite{DBLP:conf/soda/AcharyaOST15} for more details).

Under the Poisson sampling and with notations as in \Cref{alg:estimator} we arrive at the formula

\begin{align*}
 \mathbb{Var}\left[\sum_{i}q_i^{1-\alpha} \frac{n_i^{\underline{\alpha}}}{n^{\alpha}} \right] & = 
 \sum_{i}\frac{1}{q_i^{2\alpha-2}n^{2\alpha}}\mathbb{Var}\left[ \frac{n_i^{\underline{\alpha}}}{n^{\alpha}} \right] \\
 & \leqslant 
 \sum_{i}\frac{n_i^{\alpha}\left( (n_i+\alpha)^{\alpha}-n_i^{\alpha} \right)}{q_i^{2\alpha-2}n^{2\alpha}} \\
 & = \sum_{r=0}^{\alpha-1}\binom{\alpha}{r} \frac{1}{n^{\alpha-r}}\sum_{i}\frac{p_i^{\alpha+r}}{q_i^{2\alpha-2}}.
\end{align*}

The next reduction is an observation that is suffices to construct an estimator
that fails with probability at most $\frac{1}{3}$, as the success probability can be then amplified by the median trick \cite{DBLP:conf/soda/AcharyaOST15}.
In general, it is pretty standard in the literature to simply present estimators with constant error probability
\cite{DBLP:conf/stacs/CanonneDGR16}.


Let's define the success event
\begin{align*}
 failure = 1-\delta \frac{\sum_{i}q_i^{1-\alpha}}{M_{\alpha}(p,q)} \leqslant 1+\delta
\end{align*}
By Chebyszev's Inequality we obtain the following bound
\begin{align*}
\Pr\left[ failure \right] & \leqslant \delta^{-2}\frac{1}{q_j^{\alpha-1}}\sum_{r=0}^{\alpha-1}\binom{\alpha}{r}\frac{1}{n^{\alpha-r}}\left(\sum_{i}\frac{p_i^{\alpha}}{q_i^{\alpha-1}}\right)^{-\frac{\alpha-r}{\alpha}}\\
& \leqslant\delta^{-2}\frac{1}{q_j^{\alpha-1}}\sum_{r=0}^{\alpha-1}\binom{\alpha}{r}\frac{1}{n^{\alpha-r}}\left(\sum_{i}\frac{p_i^{\alpha}}{q_j^{\alpha-1}}\right)^{-\frac{\alpha-r}{\alpha}}\\ 
\end{align*}
(consistent with \cite{DBLP:conf/soda/AcharyaOST15} for uniform $q$) which finishes the proof.
\end{proof}

\section{Proof of \Cref{thm:lower_bounds}}\label{proof:thm:lower_bounds}

\begin{proof}[Proof of \Cref{thm:lower_bounds}]
We can assume that $\max(C_1,C_2)\geqslant 1$, as otherwise the estimate on $n$ is trivial.
We start with the following lemma (a similar technique
is used in \cite{DBLP:conf/soda/AcharyaOST15}, our exposition is different)
\begin{lemma}\label{lemma:distinguisher}
Suppose that there exists a $(\delta,\epsilon)$-estimator for the Renyi divergence as in \Cref{def:estimator}, 
which uses $n$ samples, where $\epsilon<\frac{1}{2}$. Then the following is true:
 any two distributions $p,p'$ that are
 $\frac{1-2\epsilon}{n}$-close in total variation, must satisfy
 $|D_{\alpha}(p\parallel q)-D_{\alpha}(p'\parallel q)| < 2\delta$.
\end{lemma}
\begin{proof}[Proof of \Cref{lemma:distinguisher}]
The lemma follows by the following observation: if the estimator fails with probability at most $\epsilon$ on both distributions
$p$ and $p'$, then one can build a distinguisher for an $n$-fold products $p^{n}$ and $p'^{n}$ by comparing the 
algorithm outputs against the threshold $\frac{1}{2}\left(D_{\alpha}(p,q)+D_{\alpha}(p',q)\right)$.
If $D_{\alpha}(p,q)-D_{\alpha}(p',q) \geqslant 2\delta$, this distinguisher works with advantage 
$1-2\epsilon$ in total variation. 
We complete the proof by the standard hybrid argument: if $n$-fold products $p^{n}$ and $p'^{n}$ are away by $1-2\epsilon$
in total variation, then the distributions $p$ and $p'$ must be $\frac{1-2\epsilon}{n}$ away.
\end{proof}

By combining this with \Cref{lemma:div_to_moments}, it suffices to prove 
that $M_{\alpha}(p',q) \geqslant (1+\Omega(1))M_{\alpha}(p,q)$ for some $p,p'$ that are 
close in total variation by $\frac{O(1)}{\max(\sqrt{C_2},C_1)}$.

Recall that $
 M_{\alpha}(p,q) = \sum_{i}{p_i^{\alpha}}{q_i^{1-\alpha}}
$. Consider any vector $\delta$ such that $\delta_i \geqslant -p_i$ and $\sum_{i}\delta_i = 0$ (in particular, $p' = p+\delta$ is a probability distribution).
By the first order Taylor approximation
\begin{align*}
 M(p+\delta,q) & \geqslant  M(p,q) + \alpha\sum_{i}\delta_i{p_i^{\alpha-1}}{q_i^{1-\alpha}}+\\ 
 & \quad + \alpha(\alpha-1)\sum_{i}\delta^2_i{(p_i+\min(0,\delta_i))^{\alpha-2}}{q_i^{1-\alpha}}.
\end{align*}
Assuming that $\delta_i \geqslant -\frac{1}{2}p_i$ we obtain
\begin{align*}
 M(p+\delta,q) & \geqslant  M(p,q) + \alpha\sum_{i}\delta_i{p_i^{\alpha-1}}{q_i^{1-\alpha}}+\\ 
 & \quad + \frac{1}{4}\alpha(\alpha-1)\sum_{i}\delta^2_i{p_i^{\alpha-2}}{q_i^{1-\alpha}}.
\end{align*}
changing variables by $\delta_i = \delta'_i p_i$, denoting $p'_i = p_i+p_i\delta'_i$ and $\delta = \sum_{i}|\delta_i|$
gives us $p'$ and $p$ that are $O(\delta)$ away in total variation and
\begin{align*}
 M(p',q) \geqslant M(p,q)(1+C_1\delta + C_2\delta^2).
\end{align*}
Consider now two cases. Assume first $C_1\geqslant 1$.
The inequality $ M(p',q)  \geqslant M(p,q)(1+C_1\delta)$ implies
an additive error of $\Omega(C_1\delta)$ in estimation.
Note that $\delta$ can be scaled (by a factor smaller than 1, as $C_1 > 1$) so that $C_1\delta = \Omega(1)$. The distance between $p$ and $p'$ is then
at least $O\left(\frac{1}{C_1}\right)$.
Suppose now that $C_2\geqslant 1$. Similarly, by scaling $\delta$ (which is possible because we have $C_2>1$) we can arrive at $C_2\delta^2 = \Omega(1)$.
Then the inequality $M(p',q)\geqslant M(p,q)(1+C_1\delta^2)$ yields an additive error $\Omega(1)$ in estimation, and 
the distance between $p$ and $p'$ is $O\left(\sqrt{\frac{1}{C_2}}\right)$. The bounds on $n$ follow, because by \Cref{lemma:distinguisher} we must have
$\frac{1}{n} < O\left(\frac{1}{C_1}\right)$ or $\frac{1}{n}< O\left(\sqrt{\frac{1}{C_2}}\right)$.
\end{proof}

\end{document}